\documentclass[11pt]{llncs}

\usepackage{fullpage}
\usepackage{xspace}
\usepackage[ruled,vlined,linesnumbered]{algorithm2e}
\usepackage{setspace}
\usepackage{mathtools}
\usepackage{amsmath}
\usepackage{amssymb}
 
\usepackage{amsfonts}
\usepackage{amsthm}
\usepackage[usenames]{color}
\usepackage{subfigure}
\usepackage[stable]{footmisc} 

\usepackage{thmtools,thm-restate}
\usepackage[bookmarks=false]{hyperref}
\usepackage{cleveref}
\usepackage{comment}
\usepackage{multicol}
\usepackage{multirow}

\usepackage{tikz}
\usepackage{pgfplots}
\tikzstyle{every node}=[circle,draw=black, inner sep=1.5pt,fill=white]
\usetikzlibrary{calc}

\newcommand{\setof}[2]{\{#1\;:\;#2\}}

\newcommand{\kmer}{$k$-mer\xspace}
\newcommand{\kmers}{$k$-mers\xspace}

\newcommand{\kpomers}{$(k+1)$-mers\xspace}

\newcommand{\pre}{{\rm pre}}
\newcommand{\suf}{{\rm suf}}
\newcommand{\spell}{{\rm spell}}

\renewcommand{\leq}{\leqslant}

\renewcommand{\geq}{\geqslant}

\newcommand{\ecoli}{{\it E.coli}\xspace}

\newcommand{\ec}{omnitig\xspace}
\newcommand{\ecs}{omnitigs\xspace}
\newcommand{\Ec}{Omnitig\xspace}

\newcommand{\model}{{\mathcal G}\xspace}
\newcommand{\sj}{Y-to-V\xspace}

\usepackage{soul} 

\declaretheorem[name=Observation]{obs}

\newcommand{\dbec}{\text{DB}_{\text{ec}}\xspace}
\newcommand{\dbnc}{\text{DB}_{\text{nc}}\xspace}

\begin{document}
\title{Safe and complete contig assembly via omnitigs\footnote{A preliminary version of this paper appeared in 
RECOMB 2016}}
\author{Alexandru I. Tomescu\inst{1} \and Paul Medvedev\inst{2,3,4}}
\institute{$^1$Helsinki Institute for Information Technology HIIT, \\
Department of Computer Science, University of Helsinki\\P.O. Box 68, FI-00014, Helsinki, Finland\\E-mail: \texttt{tomescu@cs.helsinki.fi}, Fax: +358 9 876 4314\\ 
$^2$Department of Computer Science and Engineering, The Pennsylvania State University, USA\\
$^3$Department of Biochemistry and Molecular Biology, The Pennsylvania State University, USA\\
$^4$Genome Sciences Institute of the Huck,  The Pennsylvania State University, USA}
\maketitle

\begin{abstract}
		Contig assembly is the first stage that most assemblers solve when reconstructing a genome from a set of reads. Its output consists of contigs -- a set of strings that are promised to appear in any genome that could have generated the reads. From the introduction of contigs 20 years ago, assemblers have tried to obtain longer and longer contigs, but the following question remains: given a genome graph $G$ (e.g.~a de Bruijn, or a string graph), what are \emph{all} the strings that can be safely reported from $G$ as contigs? In this paper we 
	answer this question using a model where the genome is a circular covering walk.
	We also give a polynomial time algorithm to find such strings, which we call omnitigs. 
Our experiments 
show that 
omnitigs
are 66\% to 82\% longer on average than the popular unitigs, and 
29\% of dbSNP locations have more neighbors in omnitigs than in unitigs.

\textbf{Keywords:} genome assembly, contig assembly, safe and complete algorithm, graph algorithm, omnitig

\end{abstract}

\newpage

\section{Introduction}\label{sec:intro}
The genome assembly problem is to reconstruct the sequence of a genome using reads from a sequencing experiment.
It is one of the oldest bioinformatics problems; 
nevertheless, recent projects such as the Genome 10K have underscored the need to further improve assemblers~\cite{genome10k}.
Current algorithms face numerous practical challenges, 
including scalability, 
integration of new data types (e.g. PacBio), 
and representation of multiple alleles.
While handling these challenges is extremely important, 
assemblers do not produce optimal results even in very simple and idealized scenarios.
To address this, several papers have developed better theoretical underpinnings~\cite{IW95,M05,MGMB07,simpson10,pdbg,vyahhi12},
often resulting in improved assemblers in practice~\cite{PTW01,velvet,SGA,spades}.

In most theoretical studies, the assembly problem is formulated as finding one genomic reconstruction, i.e. 
a single string that represents the sequence of the genome.
However, 
the presence of repeats means that a unique genomic reconstruction usually does not exist.
In practice, assemblers instead 
output several strings, called \emph{contigs}, that are ``promised'' to occur in the genome. 
We refer to this restatement of the genome assembly problem as \emph{contig assembly}. 
Contigs can then be used to answer biological questions (e.g. about gene content) or perform comparative genomic analysis. 
When mate pairs are available, contigs can be fed to later assembly stages, such as scaffolding~\cite{besst,sspace,soapdenovo2} and then gap filling~\cite{salmela15,boetzer12}. 

Assemblers implement different strategies for finding contigs.
The common strategy is to find {\em unitigs}, an idea that can be traced back to 1995~\cite{KM95}.
Unitigs have the desired property that they can be mathematically proven to occur in \emph{all} possible genomic reconstructions, 
under clear assumptions on what ``genomic reconstruction'' means. 
We will refer to strings that satisfy such a property as being \emph{safe} (\Cref{def:safe-string}), and 
will say that a contig assembly algorithm is {\em safe} if it outputs only safe strings.
Though most assemblers have a safe strategy at their core,
they also incorporate heuristics to handle erroneous data and extend contig length
(e.g. bubble popping, tip removal, and path disambiguation).
Properties of such heuristics, however, are difficult to prove, and 
this paper will focus on core algorithms that are safe.

While the unitig algorithm is safe, it does not identify {\em all} possible safe strings (see \Cref{fig:all-algorithms}).
An improved safe algorithm was used in the EULER assembler~\cite{PTW01},
and further improvements were suggested based on iteratively simplifying the graph used for assembly~\cite{PTW01,MB09,jacksonthesis,kingsford10}. 
However, we will show that these algorithms still do not always output all the safe strings.
In fact, since the initial consideration of contig assembly 20 years ago, the fundamental question of finding \emph{all} the safe strings of a graph remains poorly studied. 

In this paper, we answer this question by giving a polynomial-time algorithm for outputting {\em all} the safe strings in the common genome graph models (de Bruijn and string graphs) when the genome is a circular covering walk (\Cref{sec:omnitig-alg}).
The key ingredient for this result is a graph-theoretic characterization of the walks that correspond to safe strings (\Cref{sec:omnitig-char}).
We call such walks {\em omnitigs} and our algorithm the {\em \ec algorithm}.
In our experiments on de Bruijn graphs built from data simulated according to our assumptions, 
maximal omnitigs are on average $66\%$ to $82\%$ longer than maximal unitigs, and 
29\% of dbSNP locations have more neighbors in omnitigs than in unitigs.

Our results are naturally limited to the context of our model and its assumptions.
Intuitively, we assume that (i) the sequenced genome is circular, (ii) there are no gaps in coverage, and 
(iii) there are no errors in the reads.
A mathematically precise definition of our model will be presented in \Cref{sec:formulation}.
We argue that such a model is necessary if we want to 
prove even the simplest results about unitigs (\Cref{sec:formulation}).
Similar to previous studies, we also do not deal with multiple chromosomes or the double-strandedness of DNA
and assume the genome is represented by a covering walk. 
As with previous papers that developed better theoretical underpinnings~\cite{P89,IW95,M05,pdbg},
it is necessary to prove results in a somewhat idealized setting.
While this paper falls short of analyzing real data, 
we believe that omnitigs can be incorporated into practical genome analysis and assembly tools --
similar to the way that error-free studies of de Bruijn~\cite{P89} and paired de Bruijn graphs~\cite{pdbg}
became the basis of practical assemblers~\cite{PTW01,vyahhi12,spades}.






\section{Related work}\label{sec:related-work}
The number of related assembly papers is vast, and we refer the reader to some surveys~\cite{miller_assembly_2010,nagarajan2013sequence}.
For an empirical evaluation of the correctness of several state-of-the-art assemblers, see~\cite{GAGE}.
Here, we discuss work on the theoretical underpinnings of assembly.

There are many formulations of the genome assembly problem. One of the first asks to reconstruct the genome as a shortest superstring of the reads~\cite{peltola83,K92,KM95}. Later formulations referred to a graph built from the reads, such as a de Bruijn graph~\cite{IW95,PTW01} or a string graph~\cite{M05,SGA}. 
In an (edge-centric) de Bruijn graph, the reconstructed genome can be modeled as a circular walk covering every edge exactly once---Eulerian~\cite{PTW01}---or at least once---a Chinese Postman tour~\cite{MB09,MGMB07,NP09,kapun13b}. In 
a string graph, the reconstructed genome can be modeled as a circular walk covering every node exactly once---Hamiltonian---\cite{SequencingHybridizationLysov1988,narzisi14}, or at least once~\cite{NP09}. These models have also been considered in their weighted versions~\cite{MB09,NP09,narzisi14}, or augmented to include other information, such as mate-pairs~\cite{journals/jcb/RubinovG95,pdbg,kapun13a}. Each such notion of genomic reconstruction brought along questions concerning its validity. For example, under which conditions on the sequencing data (e.g., coverage, read length, error rate) is there at least one reconstruction~\cite{LW88,DBLP:journals/tit/MotahariBT13}, or exactly one reconstruction~\cite{BBT13,DBLP:journals/bmcbi/LamKT14,PTW01}. If there are many possible reconstructions, then what is their number~\cite{kingsford10,journals/bioinformatics/Guenoche92} and in which aspects one is different from all others~\cite{journals/bioinformatics/Guenoche92}. 
In contrast to the framework of this paper, most of these formulations deal with finding a single genomic reconstruction as opposed to a set of safe strings (i.e. contigs). 

There are a few notable exceptions.
In~\cite{ray}, Boisvert and colleagues also define the assembly problem in terms of finding contigs,
rather than a single reconstruction.
Nagarajan and Pop~\cite{NP09} observe that Waterman's characterization \cite{W95} of the graphs with a unique Eulerian tour leads to a simple algorithm for finding all safe strings when a genomic reconstruction is an Eulerian tour.
They also suggest an approach for finding all the safe strings when a genomic reconstruction is a Chinese Postman tour~\cite{NP09}.
We note, however, that in the Eulerian model, the exact copy count of each edge should be known in advance, while in the Chinese Postman model (minimizing the length of the genomic reconstruction), the solution will over-collapse all tandem repeats.
Furthermore, these approaches have not been implemented and hence their effectiveness is unknown.

In practice, the most commonly employed safe strings are the ones spelled by maximal \emph{unitigs}, where {\em unitigs} are paths whose internal nodes have in- and out-degree one.
\Cref{fig:all-algorithms} shows an example of the output of the unitig algorithm, and 
also illustrates that it does not identify all safe strings.
The EULER assembler~\cite{PTW01} takes unitigs a step further and identifies strings spelled by paths whose internal nodes have out-degree equal to one (with no constraint on their in-degree).  It can be shown that such strings are also safe.
However, the most complete characterization of safe strings that we found is given by the 
{\em \sj algorithm}
\cite{MGMB07,jacksonthesis,kingsford10}.
Consider a node $v$ with exactly one in-neighbor $u$ and more than one out-neighbors $w_1, \ldots, w_d$.
The {\em Y-to-V reduction} applied to $v$ removes $v$ and its incident edges from the graph and adds nodes $v_1, \ldots, v_d$ with edges from $u$ to $v_i$ and from $v_i$ to $w_i$, for all $1\leq i \leq d$.
The Y-to-V reduction is defined symmetrically for nodes with out-degree exactly one and in-degree greater than one.
\Cref{fig:Y-to-V} illustrates the definition.
The \sj algorithm proceeds by repeatedly applying Y-to-V reductions, in arbitrary order, for as long as possible. 
The algorithm then outputs the strings spelled by the maximal unitigs in the final graph (see \Cref{fig:all-algorithms}d for an example).
The Y-to-V algorithm can also be shown to be safe, but, as we will show in \Cref{fig:all-algorithms}, it does not always output {\em all} the safe strings.
We are not aware of any study that compares the merits of Y-to-V contigs to unitigs, and we therefore perform
this analysis in~\Cref{sec:experiments}.

\begin{figure}[t]
\centering
\subfigure[\label{fig:Y-to-V-out}]{
\scalebox{0.9}{
\begin{tikzpicture}[scale=1.2,->]
\node[label={above:$u$}] (u) at (0,0) {};
\node[label={above:$v$}] (v) at (1,0) {};
\node[label={above:$w_1$}] (w1) at (2,0.8) {};
\node[label={above:$w_2$}] (w2) at (2,0) {};
\node[label={below:$w_d$}] (wd) at (2,-0.8) {};
\node[draw=none,fill=none] (0) at (2,-0.32) {$\vdots$};
\draw (u) to (v);
\draw (v) to (w1);
\draw (v) to (w2);
\draw (v) to (wd);
\node[draw=none] (arrow) at (3,0) {$\Longrightarrow$};
\begin{scope}[shift={(4,0)}]
\node[label={above:$u$}] (u) at (0,0) {};
\node[label={above:$v_1$}] (v1) at (1,0.8) {};
\node[label={above:$v_2$}] (v2) at (1,0) {};
\node[label={below:$v_d$}] (vd) at (1,-0.8) {};
\node[draw=none,fill=none] (0) at (1,-0.32) {$\vdots$};
\node[label={above:$w_1$}] (w1) at (2,0.8) {};
\node[label={above:$w_2$}] (w2) at (2,0) {};
\node[label={below:$w_d$}] (wd) at (2,-0.8) {};
\node[draw=none,fill=none] (1) at (2,-0.32) {$\vdots$};
\draw (u) to (v1);
\draw (u) to (v2);
\draw (u) to (vd);
\draw (v1) to (w1);
\draw (v2) to (w2);
\draw (vd) to (wd);
\end{scope}
\end{tikzpicture}}
}
\hspace{0.3cm}
\subfigure[\label{fig:Y-to-V-in}]{
\scalebox{0.9}{
\begin{tikzpicture}[scale=1.2,->]
\node[label={above:$u$}] (u) at (2,0) {};
\node[label={above:$v$}] (v) at (1,0) {};
\node[label={above:$w_1$}] (w1) at (0,0.8) {};
\node[label={above:$w_2$}] (w2) at (0,0) {};
\node[label={below:$w_d$}] (wd) at (0,-0.8) {};
\node[draw=none,fill=none] (0) at (0,-0.32) {$\vdots$};
\draw (v) to (u);
\draw (w1) to (v);
\draw (w2) to (v);
\draw (wd) to (v);
\node[draw=none] (arrow) at (3,0) {$\Longrightarrow$};
\begin{scope}[shift={(4,0)}]
\node[label={above:$u$}] (u) at (2,0) {};
\node[label={above:$v_1$}] (v1) at (1,0.8) {};
\node[label={above:$v_2$}] (v2) at (1,0) {};
\node[label={below:$v_d$}] (vd) at (1,-0.8) {};
\node[draw=none,fill=none] (0) at (1,-0.32) {$\vdots$};
\node[label={above:$w_1$}] (w1) at (0,0.8) {};
\node[label={above:$w_2$}] (w2) at (0,0) {};
\node[label={below:$w_d$}] (wd) at (0,-0.8) {};
\node[draw=none,fill=none] (1) at (0,-0.32) {$\vdots$};
\draw (v1) to (u);
\draw (v2) to (u);
\draw (vd) to (u);
\draw (w1) to (v1);
\draw (w2) to (v2);
\draw (wd) to (vd);
\end{scope}
\end{tikzpicture}}
}
\caption{The Y-to-V reduction applied to node $v$. In Fig.~\ref{fig:Y-to-V-in} $v$ has in-degree exactly one; in Fig.~\ref{fig:Y-to-V-in} $v$ has out-degree exactly one.\label{fig:Y-to-V}}
\end{figure}
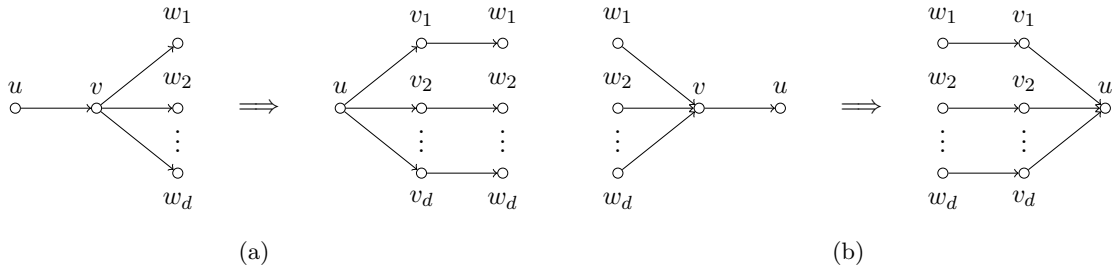


\section{Basic definitions}
Given a string $x$ and an index $1 \leq i \leq |x|$, we define $\pre(x,i)$ and $\suf(x,i)$ as its length $i$ prefix and suffix, respectively.
If $x$ and $y$ are two strings, and $\suf(x,k) = \pre(y,k)$
for some $k \leq |x| - 1$, 
then we define $x \oplus^k y$ as $x[1..|x| - k]$ concatenated with $y$.
This captures the notion of merging two overlapping strings.
A \kmer of $x$ is a substring of length $k$.
Let $R$ be a set of strings, which we equivalently refer to as {\em reads}.
The \emph{node-centric de Bruijn graph built on $R$}, denoted $\dbnc^k(R)$, is the graph whose set of nodes is the 
set of all \kmers of $R$, 
in which there is an edge from a node $x$ to a node $y$ iff $\suf(x,k-1) = \pre(y,k-1)$~\cite{minia}. 
The \emph{edge-centric de Bruijn graph built on $R$}, denoted $\dbec^k(R)$ is defined similarly to $\dbnc^k(R)$, with the difference that there is an edge from $x$ to $y$ iff $\suf(x,k-1) = \pre(y,k-1)$ \emph{and} $x \oplus^{k-1} y$ is a substring of some string in $R$~\cite{IW95}.
The \emph{weight} of the edges of $\dbnc^k(R)$ and $\dbec^k(R)$ is $k-1$.

Let $G$ be a graph, possibly with parallel edges and self-loops.
The number of nodes and edges in a graph are denoted by $n$ and $m$, respectively.
We use $N^-(v)$ to denote the set of in-neighbors and $N^+(v)$ to denote the set of out-neighbors of a node $v$.
A \emph{walk} $w$ is a sequence $(v_0,e_0,v_1,e_1,\dots,v_t,e_t,v_{t+1})$ where $v_0,\dots,v_{t+1}$ are nodes, and each $e_i$ is an edge from $v_i$ to $v_{i+1}$, and $t \geq -1$. 
Its {\em length} is its number of edges, namely $t+1$. 
A \emph{path} is a walk where the nodes are all distinct, 
except possibly the first and last nodes may be the same, in which case it will also be called a \emph{cycle}.
Walks and paths of length at least one are called \emph{proper}.
A walk whose first and last nodes coincide is called \emph{circular} walk.
A path (walk) with first node $u$ and last node $v$ will be called a \emph{path (walk) from $u$ to $v$}, and denoted as \emph{$u$-$v$ path (walk)}. 
A walk is called \emph{node-covering} if it passes through each node of $G$, and \emph{edge-covering} if it passes through each edge of $G$. 
The notions of \emph{prefix} and \emph{subwalk} are defined for walks in the natural way, e.g., by interpreting a walk to be a string made up by concatenating its edges. In particular, we say that a walk $w_1$ is a subwalk of a \emph{circular} walk $w_2$ if $w_1$ interpreted as string is a substring of $w_2$ interpreted as \emph{circular} string. 
In this paper we allow strings and walks to have overlapping extremities when viewed as substrings of a circular string, i.e., when aligned to a circular string (see e.g. the two omnitigs from Figure~\ref{fig:all-algorithms}(f) which have an overlapping tail and head).

Let $\ell$ be a function labeling the nodes  of $G$ and let $c$ be a function giving weights to the edges 
(intuitively, $c$ should represent the length of overlaps).
One can apply the notion of string spelled by a walk $w = (v_0,e_0,v_1,e_1,\dots,v_t,e_t,v_{t+1})$ by defining 
the string \emph{spelled by $w$} as
$\spell(w) = \ell(v_0) \oplus^{c(e_0)} \ell(v_1) \oplus^{c(e_1)} \cdots \oplus^{c(e_t)} \ell(v_{t+1})$. 
When the walk $w$ is circular (thus $v_{t+1} = v_0$), then $\spell(w)$ will be interpreted as the circular string obtained by overlapping the strings $\ell(v_0)$ and $\ell(v_{t+1})$.



\section{Problem formulation}\label{sec:formulation}
There are various theoretical approaches to formulating the assembly problem.
Here, we adopt a model that captures the most popular ones: 
the node-centric de Bruijn graph, 
the edge-centric de Bruijn graph,
and the string graph~\cite{M05}.
We generalize these using the notion of a {\em genome graph}:

\begin{definition}[Genome graph]
\label{def:genome-graph}
A graph $G$ with edge-weights given by $c$ and node-labels is a {\em genome graph}
if and only if 
(1) for every edge $e=(u,v)$, $\suf(u,c(e)) = \pre(v,c(e))$, and 
(2) for any two walks $w_1$ and $w_2$, $w_1$ is a subwalk of $w_2$ if and only if 
$\spell(w_1)$ is a substring of $\spell(w_2)$.
\end{definition}
Both node- and edge-centric de Bruijn graphs are genome graphs, directly by their definition.
Similarly, the interested reader can verify that string graphs, as commonly defined in~\cite{M05,NP09,MGMB07,simpson10}, are genome graphs.
Intuitively, the first condition states that the edge-weights represent the length of overlaps between strings, while the second condition prohibits a certain redundancy in the graph.
It can be broken if, for example, there are nodes with duplicate labels, or if some labels are substrings of others.
Or, for strings graphs, it can be broken if transitive edges are not removed from the graph~\cite{M05}. 
We now augment a genome graph with a rule defining a ``genomic reconstruction.'' 




\begin{definition}[Graph model]
	A \emph{graph model} ${\mathcal G}$ is defined by
	\begin{itemize}
		\item An algorithm that transforms a set of reads $R$ into a genome graph, denoted by $\model(R)$.
		\item A rule that determines whether a walk in $\model(R)$ is a {\em genomic reconstruction}.
	\end{itemize}
\end{definition}
Intuitively, a genomic reconstruction spells a genome that could have generated the observed set of reads $R$.
In this paper, we consider two graph models. In the \emph{edge-centric} model, a genomic reconstruction is a circular edge-covering walk; its underlying genome graph can be e.g.~an edge-centric de Bruijn graph. In the \emph{node-centric} model, a genomic reconstruction is a circular node-covering walk; its underlying genome graph can be a node-centric de Bruijn graph or a string graph. As mentioned in the introduction, we assume, without always explicitly stating it onwards, that $\model(R)$ contains at least one genomic reconstruction, and for technical reasons---see the proof of Lemma~\ref{thm:number-of-edge-centric-contigs}---that $\model(R)$ is always different from a single cycle. 
In fact, in this latter case the assembly problem is trivial.

We now define the strings that belong to all genomic reconstructions.

\begin{definition}[Safe string]
\label{def:safe-string}
	Given a set of reads $R$ and a graph model $\model$,
	a string $s$ is said to be a {\em safe string for $\model(R)$} if 
	for every genomic reconstruction $w$ of $\model(R)$, $s$ is a substring of $\spell(w)$.
\end{definition}
In particular, for a node-centric (respectively, edge-centric) graph model $\model$,
a string $s$ is safe if for every circular node-covering (respectively, edge-covering) walk $w$,
$s$ is a substring of $\spell(w)$. 
It also follows from the definitions (again assuming no gaps in coverage and no errors in the reads) that if the genome graph is $\dbnc^k(R)$ or $\dbec^k(R)$, then a string is safe if it is a substring of every circular string with the same set of $k$-mers, or $(k+1)$-mers respectively, as $R$.

Solving the following problem gives all the information that can be safely retrieved from a graph model. 

\begin{definition}[The safe and complete contig assembly problem]
Given a set of reads $R$ and a graph model $\model$, output \emph{all} the safe strings for $\model(R)$.
\end{definition}
In this paper we solve this problem for the node- and edge-centric models defined above. In Sections~\ref{sec:omnitig-char} and~\ref{sec:omnitig-alg} we first deal with the edge-centric model, and then in \Cref{sec:node-centric-model} we show how these results can be modified for the node-centric model.

As a technical aside, our algorithms will output only \emph{maximal} safe strings, in the sense that they are not a substring of any other safe string. 
In fact, this is desirable in practice, and moreover, the set of all safe strings is the set of all substrings of the maximal ones. 

\paragraph{A note on assumptions:}
Our model makes three implicit assumptions, as outlined at the end of the Introduction.
Here, we observe that such assumptions are necessary to prove even the simplest desired property: that the unitig algorithm outputs only safe strings.
Let $w = (v_0,e_0,v_1,e_1,v_2)$ be a unitig in an edge-centric de Bruijn graph $G$ built from the \kpomers of a genome $S$. 
If the genome is not circular (assumption (i)), then e.g.~the last $k$-mer of $S$ can be $v_0$, its first $k$-mer can be $v_1$, the string $v_0 \oplus^k v_1$ can appear inside $S$, but $v_0 \oplus^k v_1 \oplus^k v_2$ does not have to appear in $S$. 
If there are gaps in coverage (assumption (ii)), then both an in-neighbor $v'$ and an out-neighbor $v''$ of $v_1$ may be missing from $G$ making $w$ look safe whereas in reality $v_0 \oplus^k v_1 \oplus^k v_2$ may not be a substring of $S$. 
If a read contains a sequencing error (assumption (iii)), then this creates a bubble in $G$ with one of its paths being a unitig not spelling a substring of $S$.

\begin{figure}[t]
\centering

\includegraphics[scale=0.9]{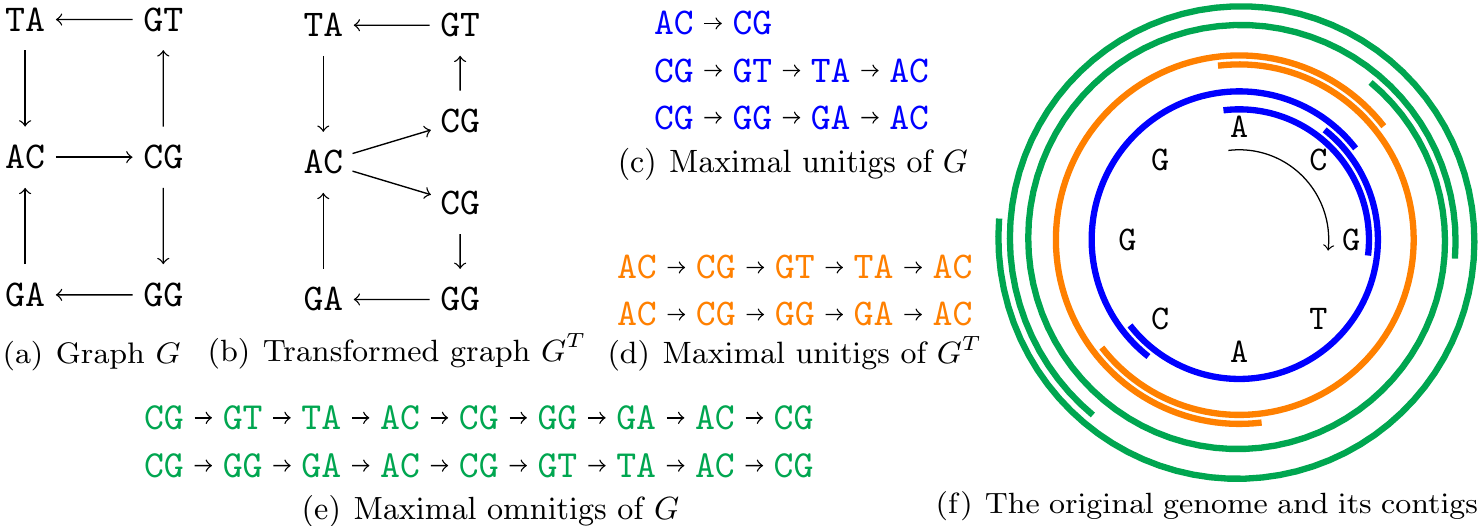}

\caption{The output of the three algorithms on the edge-centric de Bruijn graph $G$ from (a), 
built from the circular string in (f). 
Each contig is drawn as an arc on the wheel in (f). 
(c) the maximal unitigs of $G$; 
(b) the Y-to-V reduction is applied to node \texttt{CG} and the resulting graph $G^T$ is shown; no more reductions are applicable and $G^T$ has two maximal unitigs, shown in (d); 
(e) the maximal omnitigs of $G$; 
in this particular example, they are also circular edge-covering walks of $G$, and one can be obtained from the other by a circular permutation.
Note that this example illustrates that the Y-to-V algorithm does not always output all safe strings, because its output (d) does not contain the strings of (e).
\label{fig:all-algorithms}}
\end{figure}

\section{Characterization of safe strings: \ecs\label{sec:omnitig-char}}

In this section, we provide a characterization of walks that spell safe strings (see Figure~\ref{fig:omnitig-edge-centric} for an illustration). This characterization will be the basis of our omnitig algorithm in the next section.

\sloppypar
\begin{definition}[\Ec, edge-centric model]
Let $G$ be a directed graph and let $w = (v_0,e_0,v_1,e_1,\dots,v_t,e_t,v_{t+1})$ be a walk in $G$. We say that $w$ is a {\em \ec} if and only if for all $1 \leq i \leq j \leq t$, there is no proper $v_j$-$v_i$ path with first edge different from $e_j$, and last edge different from $e_{i-1}$.
\label{def:ec-omnitigs}
\end{definition}
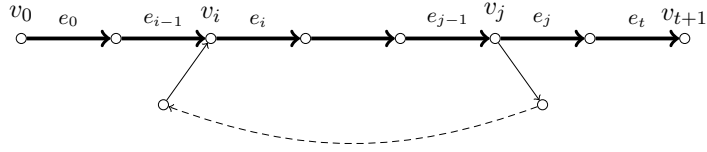
\begin{figure}
\begin{center}
\scalebox{0.9}{
\begin{tikzpicture}[scale=1.4,->]
\draw[clip,draw=none] (-.5,-1.2) rectangle (7.5,0.5);
\node[label={above:$v_0$}] (0) at (0,0) {};
\node (1) at (1,0) {};
\node[label={above:$v_i$}] (2) at (2,0) {};
\node (3) at (3,0) {};
\node (4) at (4,0) {};
\node[label={above:$v_j$}] (5) at (5,0) {};
\node (6) at (6,0) {};
\node[label={[yshift=-6pt]above:$v_{t+1}$}] (7) at (7,0) {};
\draw[ultra thick] (0) to node[draw=none,fill=none,above] {\small $e_0$} (1);
\draw[ultra thick] (1) to node[draw=none,fill=none,above,yshift=-.15cm] {\small $e_{i-1}$} (2);
\draw[ultra thick] (2) to node[draw=none,fill=none,above] {\small $e_i$} (3);
\draw[ultra thick] (3) to 
(4);
\draw[ultra thick] (4) to node[draw=none,fill=none,above,yshift=-.15cm] {\small $e_{j-1}$} (5);
\draw[ultra thick] (5) to node[draw=none,fill=none,above] {\small $e_j$} (6);
\draw[ultra thick] (6) to node[draw=none,fill=none,above] {\small $e_t$} (7);
\node (2') at (1.5,-0.7) {};\draw (2') to (2);
\node (5') at (5.5,-0.7) {};\draw (5) to (5');
\draw[densely dashed,bend left=20] (5') to (2');
\end{tikzpicture}
}
\end{center}
\caption{An illustration of the omnitig definition, edge-centric model\label{fig:omnitig-edge-centric}}
\end{figure}

The following theorem proves that the omnitigs spell all the safe strings, 
using the help of an intermediary characterization of omnitigs. 

\begin{restatable}{theorem}{thmedgecentriccontig}
Given an edge-centric graph model $G = \model(R)$ built for a set of reads $R$, and a string $s$, the following three statements are equivalent:
\begin{enumerate}
\item[\rm (1)] $s$ is a safe string for $G$;
\item[\rm (2)] $s$ is spelled by a walk $w = (v_0,e_0,v_1,e_1,\dots,v_t,e_t,v_{t+1})$ in $G$ and $w$ is an \ec;
\item[\rm (3)] $s$ is spelled by a walk $w = (v_0,e_0,v_1,e_1,\dots,v_t,e_t,v_{t+1})$ in $G$ and for all $1 \leq j \leq t$ all proper $v_j$-$v_j$ (circular)
	walks $w'$ fulfill at least one of the following conditions:
\begin{itemize}
\item[\rm (i)] the subwalk $(v_j,e_j,\dots,v_t,e_t,v_{t+1})$ of $w$ is a prefix $w'$, or
\item[\rm (ii)] the subwalk $(v_0,e_0,\dots,v_{j-1},e_{j-1},v_{j})$ of $w$ is a suffix of $w'$, or
\item[\rm (iii)] $w$ is a subwalk of $w'$.
\end{itemize}
\end{enumerate}
\label{thm:edge-centric-contig2}
\end{restatable}

We prove Theorem~\ref{thm:edge-centric-contig2} by proving the cyclical sequence of implications (1) $\Rightarrow$ (2) $\Rightarrow$ (3) $\Rightarrow$ (1). 

\begin{proof}[Proof of {\rm (1)} $\Rightarrow$ {\rm (2)}]
Assume that $s$ is a safe string for $G$. 
By definition of a genome graph, $s$ is spelled by a unique walk in $G$. 
Let $w = (v_0,e_0,v_1,e_1,\dots,v_t,e_t,v_{t+1})$ be this walk, and 
let $A$ be a circular edge-covering walk of $G$; 
thus $A$ contains $w$ as subwalk, and $s$ is a sub-string of $\spell(A)$.

Assume for a contradiction that there exist $1 \leq i \leq j \leq t$, and a proper $v_j$-$v_i$ path $p$ with first edge different from $e_j$ and last edge different from $e_{i-1}$. 
From $A$, we will construct another circular edge-covering walk $B$ of $G$ which does not contain $w$ as subwalk,
and hence, by the definition of a genome graph, also $\spell(B)$ does not contain $s$ as sub-string.
This will contradict the safeness of $s$.
Whenever $A$ visits node $v_j$, then $B$ follows the $v_j$-$v_i$ path $p$, then follows $(v_i,e_i,\dots,e_{j-1},v_j)$, and finally continues as $A$. To see that $w$ does not appear as a subwalk of $B$, consider the subwalk $w' = (v_{i-1},e_{i-1},v_i,e_i,\dots,e_{j-1},v_j,e_j,v_{j+1})$ of $w$ (recall that $1 \leq i \leq j \leq t$). Since $p$ is proper, and its first edge is different from $e_j$ and its last edge is different from $e_{i-1}$, then, by construction, the only way that $w'$ can appear in $B$ is as a subwalk of $p$. However, this implies that both $v_j$ and $v_i$ appear twice on $p$, contradicting the fact that $p$ is a 
path.
\end{proof}

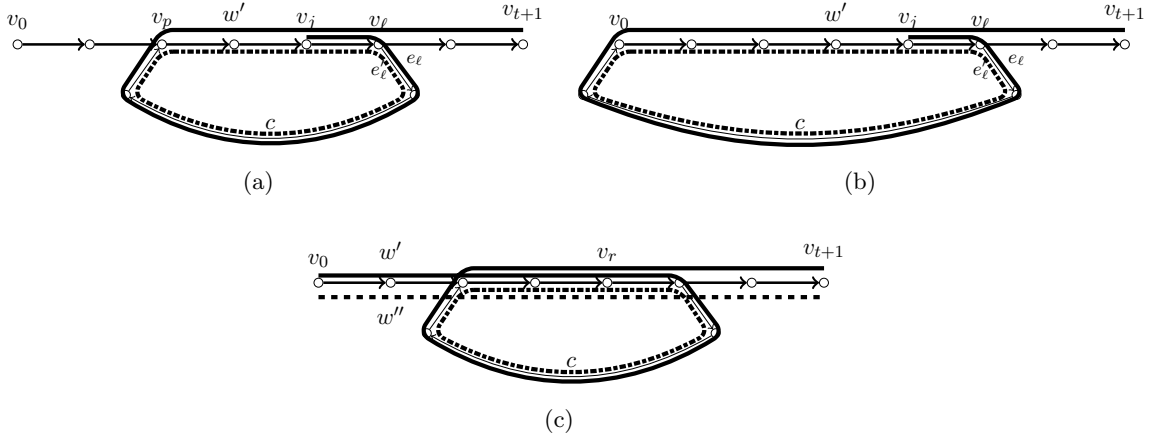
\begin{figure}[h]
\centering
\subfigure[\label{fig:edge-centric-contig-charact-1a}]{\scalebox{0.8}
{\begin{tikzpicture}[scale=1.2,->]
\draw[clip,draw=none] (-.7,-1.5) rectangle (7.5,0.6);
\node[label={above:$v_0$}] (0) at (0,0) {};
\node (1) at (1,0) {};
\node[label={above:$v_p$}] (2) at (2,0) {};
\node (3) at (3,0) {};
\node[label={above:$v_j$}] (4) at (4,0) {};
\node[label={above:$v_\ell$}] (5) at (5,0) {};
\node (6) at (6,0) {};
\node[label={above:$v_{t+1}$}] (7) at (7,0) {};
\draw[very thick] (0) to (1);
\draw[very thick] (1) to (2);
\draw[very thick] (2) to (3);
\draw[very thick] (3) to (4);
\draw[very thick] (4) to (5);
\draw[very thick] (5) to node[draw=none,fill=none,below] {\small $e_\ell$} (6);
\draw[very thick] (6) to (7);
\node (2') at (1.5,-0.7) {};\draw (2') to (2);
\node (5') at (5.5,-0.7) {};\draw (5) to node[draw=none,fill=none,left] {\small $e_\ell'$} (5');
\draw[bend left=30] (5') to (2');

\draw[-,rounded corners=5pt,draw=black,line width=2pt,densely dashdotted]
($ (2) + (0,-0.1) $) -- ($ (5) + (0,-0.1) $) -- ($ (5') + (-0.1,0) $) to[bend left=28] ($ (2') + (0.1,0) $) -- cycle;
\draw[-,rounded corners=5pt,draw=black,line width=2pt,solid]
($ (4) + (0,0.1) $) to ($ (5) + (0,0.1) $) to ($ (5') + (0.1,0) $) to[bend left=32] ($(2') + (-0.1,0)$) to ($ (2) + (0,0.2) $) to ($ (7) + (0,0.2) $);
\node[draw=none,fill=none] at (3.5,-1.1) {$c$};
\node[draw=none,fill=none] at (3,.45) {$w'$};
\end{tikzpicture}}}
\medskip
\subfigure[\label{fig:edge-centric-contig-charact-1b}]{\scalebox{0.8}
{\begin{tikzpicture}[scale=1.2,->]
\draw[clip,draw=none] (-.7,-1.5) rectangle (7.5,0.6);
\node[label={above:$v_0$}] (0) at (0,0) {};
\node (1) at (1,0) {};
\node (2) at (2,0) {};
\node (3) at (3,0) {};
\node[label={above:$v_j$}] (4) at (4,0) {};
\node[label={above:$v_\ell$}] (5) at (5,0) {};
\node (6) at (6,0) {};
\node[label={above:$v_{t+1}$}] (7) at (7,0) {};
\draw[very thick] (0) to (1);
\draw[very thick] (1) to (2);
\draw[very thick] (2) to (3);
\draw[very thick] (3) to (4);
\draw[very thick] (4) to (5);
\draw[very thick] (5) to node[draw=none,fill=none,below] {\small $e_\ell$} (6);
\draw[very thick] (6) to (7);
\node (0') at (-0.5,-0.7) {};\draw (0') to (0);
\node (5') at (5.5,-0.7) {};\draw (5) to node[draw=none,fill=none,left] {\small $e_\ell'$} (5');
\draw[bend left=20] (5') to (0');

\node[draw=none,fill=none] at (2.5,-1.1) {$c$};
\node[draw=none,fill=none] at (3,.45) {$w'$};

\draw[-,rounded corners=5pt,draw=black,line width=2pt,densely dashdotted]
($ (0) + (0,-0.1) $) -- ($ (5) + (0,-0.1) $) -- ($ (5') + (-0.1,0) $) to[bend left=18] ($ (0') + (0.1,0) $) -- cycle;
\draw[-,rounded corners=5pt,draw=black,line width=2pt,solid]
($ (4) + (0,0.1) $) to ($ (5) + (0,0.1) $) to ($ (5') + (0.1,0) $) to[bend left=22] ($(0') + (-0.1,0)$) to ($ (0) + (0,0.2) $) to ($ (7) + (0,0.2) $);

\end{tikzpicture}}}
\subfigure[\label{fig:edge-centric-contig-charact-1c}]{\scalebox{0.8}
{\begin{tikzpicture}[scale=1.2,->]
\draw[clip,draw=none] (-.7,-1.5) rectangle (7.5,0.6);
\node[label={above:$v_0$}] (0) at (0,0) {};
\node (1) at (1,0) {};
\node (2) at (2,0) {};
\node (3) at (3,0) {};
\node[label={[yshift=3pt]above:$v_r$}] (4) at (4,0) {};
\node (5) at (5,0) {};
\node (6) at (6,0) {};
\node[label={above:$v_{t+1}$}] (7) at (7,0) {};
\draw[very thick] (0) to (1);
\draw[very thick] (1) to (2);
\draw[very thick] (2) to (3);
\draw[very thick] (3) to (4);
\draw[very thick] (4) to (5);
\draw[very thick] (5) to (6);
\draw[very thick] (6) to (7);
\node (2') at (1.5,-0.7) {};\draw (2') to (2);
\node (5') at (5.5,-0.7) {};\draw (5) to (5');
\draw[bend left=30] (5') to (2');
\draw[-,rounded corners=5pt,draw=black,line width=2pt,densely dashdotted]
($ (2) + (0,-0.1) $) -- ($ (5) + (0,-0.1) $) -- ($ (5') + (-0.1,0) $) to[bend left=28] ($ (2') + (0.1,0) $) -- cycle;
\draw[-,rounded corners=5pt,draw=black,line width=2pt,solid]
($ (0) + (0,0.1) $) to ($ (5) + (0,0.1) $) to ($ (5') + (0.1,0) $) to[bend left=32] ($(2') + (-0.1,0)$) to ($ (2) + (0,0.2) $) to ($ (7) + (0,0.2) $);
\draw[-,rounded corners=5pt,draw=black,line width=2pt,dashed]
($ (0) + (0,-0.2) $) to ($ (7) + (0,-0.2) $);
\node[draw=none,fill=none] at (3.5,-1.1) {$c$};
\node[draw=none,fill=none] at (1,.45) {$w'$};
\node[draw=none,fill=none] at (1,-.45) {$w''$};
\end{tikzpicture}}}
\caption{Illustration of three cases in the proof of the implication (2) $\Rightarrow$ (3) of \Cref{thm:edge-centric-contig2}.\label{fig:edge-centric-contig-charact-1}}
\end{figure}
\begin{proof}[Proof of {\rm (2)} $\Rightarrow$ {\rm (3)}]
Suppose that $w$ is an \ec, and assume for a contradiction that there exists a proper $v_j$-$v_j$ walk (for some $1 \leq j \leq t$) not satisfying (i)--(iii). 
Let $w'$ be the shortest such walk.
Since $w'$ does not have $(v_j,e_j,\dots,v_t,e_t,v_{t+1})$ as prefix, then there exists a first node $v_\ell$ on $w$, $j \leq \ell \leq t$, such that from $v_\ell$, $w'$ continues with an edge $e_{\ell}' \neq e_{\ell}$. Symmetrically, since $w'$ does not have $(v_0,e_0,\dots,v_{j-1},e_{j-1},v_{j})$ as suffix, let $v_i$ be the last node of $w$, $1 \leq i \leq j$ such that before entering $v_i$, the walk $w'$ uses an edge $e_{i-1}' \neq e_{i-1}$.
Let $w'_0$ denote the subwalk of $w'$ between $e'_{\ell}$ and $e'_{i-1}$ (inclusive). 
If $w'_0$ is a path, then $w''$ is a proper $v_{\ell}$-$v_i$ path, $1 \leq i \leq \ell \leq t$, whose first edge $e_{\ell}'$ is different from $e_\ell$, and its last edge $e_{i-1}'$ is different from $e_{i-1}$, which contradicts the fact that $w$ is an \ec.
We now prove that $w'$ is in fact a path.

	Suppose for a contradiction that it is not, thus that it contains a cycle $c$, with $c \neq w'$. Let $w''$ be the 
walk obtained from $w'$ by removing the cycle $c$. Observe that $w''$ is still a proper $v_j$-$v_j$ walk. We show that $w''$ still does not satisfy (i)--(iii), which will contradict the minimality of $w'$. Assume for a contradiction that $w''$ satisfies at least one of (i), (ii), or (iii). 

First, if $w''$ satisfies (i), this implies that the edge $e_\ell'$ out-going from $v_{\ell}$ belongs to $c$, and after traversing $c$, the walk $w'$ continues through $(v_\ell,e_\ell,\dots,v_t,e_t,v_{t+1})$ (see Figures~\ref{fig:edge-centric-contig-charact-1a} and~\ref{fig:edge-centric-contig-charact-1b}). 
Let $v_p$ be the node of $w$ with greatest index $p \in \{0,\dots,\ell\}$ that $c$ visits with an edge $e' = (v,v_p)$ not on $w$. Such a node exists because $c$ is a cycle and it must return to $v_\ell$. If $p \geq 1$ (see \Cref{fig:edge-centric-contig-charact-1a}), then $c$ does not satisfy (i)--(iii). Since $c$ is proper and passes through $v_\ell$, where $1 \leq \ell \leq t$, this contradicts the minimality of $w'$. Therefore, $p = 0$ (see \Cref{fig:edge-centric-contig-charact-1b}), and thus, the initial $v_j$-$v_j$ walk $w'$ (containing $c$ as subwalk) visits $(v_0,e_0,\dots,v_\ell)$, and then continues through $(v_\ell,e_\ell,\dots,v_t,e_t,v_{t+1})$. This implies that $w'$ contains $w$ as subwalk, which contradicts the choice of $w'$.

The second case when $w''$ satisfies (ii) is entirely symmetric.

Third, assume that $w''$ contains $w$ as subwalk. Since $w$ is not a subwalk of $w'$, this implies that $c$ is a proper $v_r$-$v_r$ walk, for some $1 \leq r \leq t$, not satisfying (i)--(iii), which again contradicts the minimality of $w'$ (see \Cref{fig:edge-centric-contig-charact-1c}).
This completes the proof of {\rm (2)} $\Rightarrow$ {\rm (3)}
\end{proof}

\begin{proof}[Proof of {\rm (3)} $\Rightarrow$ {\rm (1)}]
Assume $w$ satisfies (3), and let $A$ be a circular edge-covering walk of $G$. We need to show that $w$ is a subwalk of $A$. Let $w_j = (v_0,e_0,\dots,v_{j-1},e_{j-1},v_j)$ be the longest prefix of $w$ that $A$ ever traverses, ending at some $v_j$. Since $A$ covers all edges, then it also covers $e_0$, and thus $j \geq 1$. Suppose for a contradiction that $j \neq t+1$. 

Since $A$ is circular and covers all edges of $G$, then after traversing $w_j$, the walk $A$ eventually visits the edge $e_j$. The walk $A$ may visit $v_j$ multiple times before traversing the edge $e_j$. Let $w'$ denote the subwalk of $A$ between the last two occurrences of $v_j$ before $A$ traverses the edge $e_j$. Since $w'$ is a proper $v_j$-$v_j$ walk, $1 \leq j \leq t$, and $w$ satisfies (3), we have that one of the following must hold:
\begin{itemize}
\item the walk $(v_j,e_j,\dots,v_t,e_t,v_{t+1})$ is a prefix of $w'$: this contradicts the fact that $w'$ is a subwalk of $A$ between $v_j$ and the immediately next occurrence of $e_j$ (since in this case $w'$ contains $e_j$);
\item the walk $(v_0,e_0,\dots,v_{j-1},e_{j-1},v_{j})$ is a suffix of $w'$: this implies that $(w',e_j,v_{j+1})$ is a longer prefix of $w$ which is a subwalk of $A$, contradicting the maximality of $w_j$;
\item the walk $w$ appears on $w'$: since $w'$ is a subwalk of $A$, this implies that also $w$ is a subwalk of $A$, contradicting again the maximality of $w_j$.
\end{itemize}

\end{proof}

\section{Omnitig algorithm\label{sec:omnitig-alg}}
\begin{algorithm}[t]
\small
\caption{Omnitig algorithm to find all safe strings of a graph $G$.\label{alg:complete-contigs}}

\SetKwBlock{extend}{${\bf extend}(w)$}{end}

\extend{

Denote $w = (v_0,e_0,v_1,e_1,\dots,v_{t-1},e_{t-1},v_t)$; 

\ForEach{edge $e = (v_t,y)$ out-going from $v_t$}
{
	$X := (N^-(v_1) \setminus \{v_0\}) \cup \cdots \cup (N^-(v_{t}) \setminus \{v_{t-1}\})$\;
	let $G'$ equal $G$ minus the edge $e$\;
	\If{there is no path in $G'$ from $v_t$ to a node of $X$}
	{
		{\bf extend}($(v_0,e_0,v_1,e_1,\dots,v_{t-1},e_{t-1},v_t,e,y)$)\;
	}

}
\If{$w$ was never extended}
	{
			$W := W \cup \{w\}$\;	
	}
}

\medskip
$W := \emptyset$\; 

\ForEach{edge $e = (u,v)$ of $G$}
{
	{\bf extend}($(u,e,v)$)\;
}

remove from $W$ any walk that is a subwalk of another walk in $W$\; 
\Return $\setof{\spell(w)}{w\in W}$\;
\end{algorithm}

In this section, we use \Cref{thm:edge-centric-contig2} to give the omnitig algorithm (\Cref{alg:complete-contigs}) 
and prove that it runs in polynomial time 
(\Cref{thm:nc-ss-algorithm}).
The algorithm finds all maximal omnitigs of $\model(R)$, which, by \Cref{thm:edge-centric-contig2}, are 
exactly the maximal safe strings of $\model(R)$.
Our algorithm is based on the following observation, which follows directly from the definition of omnitigs:
\begin{obs}\label{obs:extend}
	Consider a walk $w' = (v_0,e_0,\dots,e_{t-1},v_t, e_t, v_{t+1})$ of length at least two,
	and consider its subwalk $w=(v_0,e_0,\dots,e_{t-1},v_t)$.
	Then $w'$ is an omnitig if and only if (i) $w$ is an omnitig and 
	(ii) for all $0 \leq i \leq t - 1$, there is no proper $v_t$-$v_i$ path with first edge different from $e_t$
	and last edge different from $e_{i-1}$. 
\end{obs}

The idea of the algorithm is to start an exhaustive traversal of $G$ from every edge (Lines 11-12), 
which by definition is an omnitig, and to keep traversing edges as long as the current walk is an omnitig. 
An omnitig $w$ is thus recursively constructed,
by possibly extending to the right with each edge $e$ out-going from its last vertex (Lines 3-7). 
If $w$ extended with $e$ is not an omnitig, then we abandon this extension because Observation~\ref{obs:extend} 
tells us that no further extension could be an omnitig. To check if this extension is an omnitig or not, it is enough to check whether condition (ii) of~Observation~\ref{obs:extend} is satisfied. 
Condition (i) is automatically satisfied because of the structure of the algorithm--we extend only walks that are omnitigs.
The omnitigs found are saved in a set $W$ (Line 9), except for those omnitigs that are obviously non-maximal (Line 8).
In the final step (Lines 13-14), we remove the non-maximal omnitigs from $W$ and report the rest.

To check that condition (ii) is satisfied (Lines 4-6),
we take the set $X$ (Line 4) and check 
if there is a path starting with an edge out-going from $v_t$ and different from $e$, and leading to a node of $X$. 
The correctness of this procedure can be seen as follows. 
If there is no such path, then we know that there is no path satisfying 
(ii).
If we do find a path $p$ from $v_t$ to some in-neighbor $x \in X$ of some $v_i$, and $p$ does not use $v_i$, then 
the path obtained by extending $p$ to $v_i$ contradicts 
(ii).
If $p$ contains $v_i$, then such an extension is not possible, because a path cannot repeat a vertex; however, 
we will show that $p$ cannot use $v_i$ by contradiction.
Assume that it does, and observe that after passing through $v_i$, 
the path $p$ cannot pass again through $v_t$. Let $v_j$, $i \leq j < t$, be the first vertex that $p$ visits after $v_i$ such that from $v_j$ it continues with an edge $e' \neq e_j$. Let $p'$ denote the $v_j$-$x$ subpath of $p$ from $v_j$ until $x$. We obtained that $p'$ followed by $v_i$ is a proper $v_j$-$v_i$ path with first edge different from $e_j$, last edge different from $e_{i-1}$, and $1 \leq i \leq j \leq t$.
This contradicts the fact that $w$ (the walk we are extending) is an omnitig.

Next, we show that the algorithm runs in polynomial time.
First, we show that the number of omnitigs included in $W$ and their length, prior to removal of non-maximal ones, is polynomial:
\begin{restatable}{lemma}{thmnumberofedgecentriccontigs}\label{thm:number-of-edge-centric-contigs}
	Let $W$ be a set of omnitigs in an edge-centric graph model $\model(R)$, whose genome graph is different than a single cycle.
	Furthermore, suppose no omnitig in $W$ is a prefix of another omnitig in $W$.
	Then, $|W| \leq nm$ and the length of any omnitig in $W$ is $O(nm)$. 
\end{restatable}

\begin{proof}
We first show that we can visit the edges of $G = \mathcal{G}(R)$ with a circular edge-covering walk $C$ of at most $nm$ nodes. Let $e_0,\dots,e_{m-1}$ be an arbitrary order of the edges of $G$. Since we assume that $G$ admits one genomic reconstruction, then $G$ is strongly connected. Thus, from every end extremity of $e_i$ there is a path to the start extremity of $e_{(i+1) \bmod m}$, $0 \leq i \leq m-1$, of length at most $n-1$. Therefore, $C$ can be constructed to first visit $e_0$, then to follow such a path until $e_1$, and so on until $e_{m-1}$, from where it follows such a path back to $e_0$.

By \Cref{thm:edge-centric-contig2}, we have that any \ec of $G$ is a subwalk of $C$. 
We can associate every $w\in W$ with all the start positions in $C$ (in terms of nodes) where it is a subwalk.
Because $W$ does not contain walks that are prefixes of other walks, 
a position of $C$ can have at most one walk associated with it.
Since $|C| \leq nm$, $W$ can contain at most $nm$ walks.

It remains to prove that the length of any omnitig in $W$ is $O(nm)$. To simplify notation, rename $C$ as $(v_0,e_0,v_1,e_1,\dots,v_t,e_t,v_{t+1})$ with $v_{t+1} = v_0$. Since $G$ is different than a single cycle, then there exist $v_j$ and $v_i$ on $C$, such that $e = (v_j,v_i)$ is an edge of $G$, and $e \notin \{e_{j},e_{i-1}\}$. Any omnitig (thus a subwalk of $C$) cannot contain twice $v_j$ and $v_i$ as internal nodes, since otherwise the proper path $(v_j,e,v_i)$ violates the omnitig definition. Thus the length of any omnitig is $O(nm)$.
\end{proof}

(As an aside, it remains open whether the bound on $|W|$ can be reduced to $m$, which is the case for unitigs; 
our experiments from \Cref{sec:experiments} suggest this may be the case in practice.)
Note that Line 8 guarantees that $W$, prior to removal of subwalks in Line 13, satisfies the prefix condition of
\Cref{thm:number-of-edge-centric-contigs}.
\Cref{thm:number-of-edge-centric-contigs} then implies that reporting one \ec by our algorithm takes polynomial time, and there are only polynomially many omnitigs reported. 
Furthermore, removing the non-maximal omnitigs (Line 13) can be done in linear time in the sum of the omnitig lengths, by appropriately traversing a suffix tree constructed from them.
Thus, we have our main theorem:
\begin{restatable}{theorem}{thmncssalgorithm}
Let $R$ be a set of reads and let $\model(R)$ be an edge-centric graph model. \Cref{alg:complete-contigs} outputs in polynomial time all safe strings of $\model(R)$.
\label{thm:nc-ss-algorithm}
\end{restatable}

Finally, we note some implementation details that are crucial in practice.
Prior to starting, we apply the Y-to-V algorithm and the standard graph compaction algorithm to compact 
unitigs~\cite{bcalm}.
This significantly reduces the number of nodes/edges in the graph 
without changing the maximal safe strings.
We also precompute all \ecs of length two and store them in a hash table, so that 
whenever we want to extend the \ec $w$ in Line 6, we check beforehand whether the pair $(e_{t-1},e)$ is stored in the hash table. 
This significantly limits, in practice, the number of graph traversals we have to do at Line 6.
Finally, we do not compute the set $X$ every time, but instead incrementally built it up as we extend the omnitig $w$.  Our implementation is freely available for use\footnote{\tt\url{https://github.com/alexandrutomescu/complete-contigs}}.

\section{Node-centric model}\label{sec:node-centric-model}
In this section we obtain analogous results for node-centric models, though both the definitions, algorithms, and proofs need to modified. 
The following definition is similar to the one for the edge-centric model, the only addition being its second bullet 
(see Figure~\ref{fig:omnitig-node-centric} for an illustration).

\begin{definition}[\Ec, node-centric model]
Let $G$ be a directed graph and let $w = (v_0,e_0,v_1,e_1,\dots,v_t,e_t,v_{t+1})$ be a walk in $G$. We say that $w$ is a {\em \ec} iff the following two conditions hold:
\begin{itemize}
\item for all $1 \leq i \leq j \leq t$, there is no proper $v_j$-$v_i$ path with first arc different from $e_j$, and last arc different from $e_{i-1}$.
\item for all $0 \leq j \leq t$, the arc $e_j$ is the only $v_j$-$v_{j+1}$ path.
\end{itemize}
\end{definition}

\begin{figure}
\begin{center}
\scalebox{0.9}{
\begin{tikzpicture}[scale=1.4,->]
\draw[clip,draw=none] (-.5,-1.2) rectangle (7.5,0.5);
\node[label={above:$v_0$}] (0) at (0,0) {};
\node (1) at (1,0) {};
\node[label={above:$v_i$}] (2) at (2,0) {};
\node (3) at (3,0) {};
\node (4) at (4,0) {};
\node[label={above:$v_j$}] (5) at (5,0) {};
\node (6) at (6,0) {};
\node[label={[yshift=-6pt]above:$v_{t+1}$}] (7) at (7,0) {};
\draw[ultra thick] (0) to node[draw=none,fill=none,above] {\small $e_0$} (1);
\draw[ultra thick] (1) to node[draw=none,fill=none,above,yshift=-.15cm] {\small $e_{i-1}$} (2);
\draw[ultra thick] (2) to node[draw=none,fill=none,above] {\small $e_i$} (3);
\draw[ultra thick] (3) to 
(4);
\draw[ultra thick] (4) to node[draw=none,fill=none,above,yshift=-.15cm] {\small $e_{j-1}$} (5);
\draw[ultra thick] (5) to node[draw=none,fill=none,above] {\small $e_j$} (6);
\draw[ultra thick] (6) to node[draw=none,fill=none,above] {\small $e_t$} (7);
\node (2') at (1.5,-0.7) {};\draw (2') to (2);
\node (5') at (5.5,-0.7) {};\draw (5) to (5');
\draw[densely dashed,bend left=20] (5') to (2');
\draw[densely dashed,bend right=50] (5) to (6);
\end{tikzpicture}}
\end{center}
\caption{An illustration of the omnitig definition, node-centric model\label{fig:omnitig-node-centric}}
\end{figure}
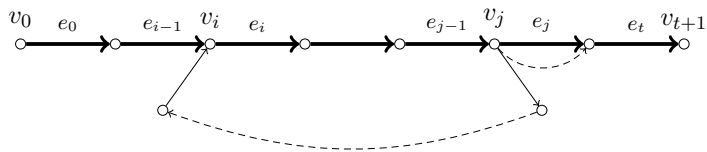

The following theorem is analogous to \Cref{thm:edge-centric-contig2}, and characterizes the safe strings in the node-centric model.

\begin{theorem}
Given a node-centric graph model $G = \model(R)$ built for a set of reads $R$, and a string $s$, the following three statements are equivalent:
\begin{enumerate}
\item[\rm (1)] $s$ is a safe string for $G$;
\item[\rm (2)] $s$ is spelled by a walk $w = (v_0,e_0,v_1,e_1\dots,v_t,e_t,v_{t+1})$ in $G$ and $w$ is a $\ec$;
\item[\rm (3)] $s$ is spelled by a walk $w = (v_0,e_0,v_1,e_1\dots,v_t,e_t,v_{t+1})$ in $G$ and $w$ satisfies: for all $0 \leq j \leq t$, the arc $e_j$ of $w$ is the only $v_j$-$v_{j+1}$ path, and for all $1 \leq j \leq t$, all proper $v_j$-$v_j$ (circular) walks $w'$ fulfill at least one of the following conditions:
\begin{itemize}
\item[\rm (i)] the subwalk $(v_j,e_j,\dots,v_t,e_t,v_{t+1})$ of $w$ is a prefix $w'$, or
\item[\rm (ii)] the subwalk $(v_0,e_0,\dots,v_{j-1},e_{j-1},v_{j})$ of $w$ is a suffix of $w'$, or
\item[\rm (iii)] $w$ is a subwalk of $w'$.
\end{itemize}
\end{enumerate}
\label{thm:node-centric-contig}
\end{theorem}

We analogously prove Theorem~\ref{thm:node-centric-contig} by proving the cyclical sequence of implications (1) $\Rightarrow$ (2) $\Rightarrow$ (3) $\Rightarrow$ (1).

\begin{proof}[Proof of {\rm (1)} $\Rightarrow$ {\rm (2)}]
Assume that $s$ is a safe string for $R$. 
By definition of a graph model, a safe string for $R$ is spelled by a unique walk in a node-centric model for $R$. 
Let this walk be $w$, and let $A$ be a circular node-covering walk of $G$ (thus containing $w$ as subwalk).

First, assume for a contradiction that there exist $1 \leq i \leq j \leq t$, and a proper $v_j$-$v_i$ path $p$ whose first arc is different from $e_j$ and its last arc is different from $e_{i-1}$. From $A$, we can construct another circular node-covering walk $B$ of $G$ which does not contain $w$ as subwalk, and thus $\spell(B)$ does not contain $s$ as sub-string. This will contradict the fact that $s$ is a safe string for $G$. 

Whenever $A$ visits node $v_j$, then $B$ follows the $v_j$-$v_i$ path $p$, then it follows $(v_i,e_i,\dots,e_{j-1},v_j)$, and finally continues as $A$. To see that $w$ does not appear as a subwalk of $B$, consider the subwalk $w' = (v_{i-1},e_{i-1},v_i,e_i,\dots,e_{j-1},v_j,e_j,v_{j+1})$ of $w$ (recall that $1 \leq i \leq j \leq t$). Since $p$ is proper, and its first arc is different from $e_j$ and its last arc is different from $e_{i-1}$, then, by construction, the only way that $w'$ can appear in $B$ is as a subwalk of $p$. However, this implies that both $v_j$ and $v_i$ appear twice on $p$, contradicting the fact that $p$ is a $v_j$-$v_i$ path.

Second, assume for a contradiction that there is some $0 \leq j \leq t$ and another $v_j$-$v_{j+1}$ path $p'$ than the arc $e_j$. Just as above, from $A$ we can construct another node-covering walk $C$ that avoids $w$ (and thus $\spell(C)$ does not contain $s$ as sub-string) as follows. Whenever $A$ traverses the arc $e_j$, $C$ traverses instead $p'$. The walk $C$ is node-covering because it still covers all nodes of $w$, and otherwise $C$ coincides with $A$. However, it does not contain $w$ as subwalk because $p'$ is different from the arc $e_j$, and, since it is a path, it cannot pass through $e_j$ again, as otherwise it would visit twice either $v_j$, $v_{j+1}$, or both.
\end{proof}

The proof of {\rm (2)} $\Rightarrow$ {\rm (3)} is identical to the corresponding proof of {\rm (2)} $\Rightarrow$ {\rm (3)} for \Cref{thm:edge-centric-contig2}.

\begin{proof}[Proof of {\rm (3)} $\Rightarrow$ {\rm (1)}]
Assume $w$ satisfies (3), and let $A$ be a circular node-covering walk of $G$. We need to show that $w$ is a subwalk of $A$. Let $w_j = (v_0,e_0,\dots,v_{j-1},e_{j-1},v_j)$ be the longest prefix of $w$ that $A$ ever traverses, ending at some $v_j$. Since $A$ covers all nodes, then $j \geq 0$. Suppose for a contradiction that $j \neq t+1$. Since $A$ is circular and covers all nodes of $G$, then after traversing $w_j$, the walk $A$ eventually visits the node $v_{j+1}$. The walk $A$ may, or may not, visit $v_j$ again before visiting the node $v_{j+1}$. 

First, suppose that after visiting $v_j$ at the end of $w_j$, $A$ visits again $v_j$ before visiting $v_{j+1}$. Let $w'$ denote the subwalk of $A$ between the last two occurrences of $v_j$ before visiting the node $v_{j+1}$. If $1 \leq j \leq t$, since $w'$ is a $v_j$-$v_j$ walk, and $w$ satisfies (3), we have that either:
\begin{itemize}
\item the walk $(v_j,e_j,v_{j+1},\dots,v_t,e_t,v_{t+1})$ is a prefix of $w'$: this contradicts the fact that $w'$ is a subwalk of $A$ between $v_j$ and the immediately next occurrence of $v_{j+1}$, since in this case $w'$ would contain $v_{j+1}$ more times;
\item the walk $(v_0,e_0,\dots,v_{j-1},e_{j-1},v_{j})$ is a suffix of $w'$: this implies that $(w',e_j,v_{j+1})$ is a longer prefix of $w$ that is a subwalk of $A$, contradicting the maximality of $w_j$;
\item the walk $w$ appears on $w'$: since $w'$ is a subwalk of $A$, this implies that also $w$ is a subwalk of $A$, contradicting again the maximality of $w_j$.
\end{itemize}
If $j=0$, then by removing all cycles from $w'$, we obtain a $v_j$-$v_{j+1}$ path, different than the arc $e_0$, since otherwise we would contradict the maximality of $w_j$. But this contradicts the fact that $w$ is satisfies (3).

Second, suppose that the walk $A$ does not visit $v_j$ again after $w_j$ and before visiting $v_{j+1}$. Let $w''$ be the $v_j$-$v_{j+1}$ subwalk of $A$ between $w_j$ and this next occurrence of $v_{j+1}$. The walk $w''$ may not be a path, but by removing all cycles from it we obtain a $v_j$-$v_{j+1}$ path $w'''$. This path is different from $e_j$ by the maximality of $w_j$, contradicting again the fact that $w$ satisfies (3).
\end{proof}

Analogous to the edge-centric case, we can prove the following polynomial upper-bound on the number and length of all omnitigs.

\begin{lemma}\label{thm:number-of-node-centric-contigs}
	Let $W$ be a set of omnitigs in a node-centric graph model $\model(R)$, whose genome graph is different than a single cycle.
	Furthermore, suppose no omnitig in $W$ is a prefix of another omnitig in $W$.
	Then, $|W| \leq n^2$ and the length of any omnitig in $W$ is $O(n^2)$.
\end{lemma}


We also leave open the question whether the bound on the number of maximal omnitigs in the node-centric model can be reduced to $n$. We now combine \Cref{thm:node-centric-contig} and \Cref{thm:number-of-node-centric-contigs} for obtaining our polynomial-time safe and complete assembly algorithm.

\begin{restatable}{theorem}{thmncssalgorithmec}
There is a safe and complete assembly algorithm for any node-centric graph model $\model(R)$ built on a set $R$ of reads, which runs in polynomial time.
\label{thm:nc-ss-algorithm-nc}
\end{restatable}

\begin{proof}
The proof is identical to the one for the edge-centric case (\Cref{alg:complete-contigs} and \Cref{thm:nc-ss-algorithm}). The only difference that needs to be made to \Cref{alg:complete-contigs} is to check that the second bullet in the definition of \ec for the node-centric case holds. This can be similarly performed by a single graph traversal, and only for the last edge added to the omnitig.
\end{proof}

\section{Experimental results\label{sec:experiments}}
We wanted to test the potential of omnitigs as an alternative to unitigs,
under the assumptions of \Cref{sec:formulation}.
We chose two genomes: one bacterial genome, \ecoli, and one larger genome, Human chr10 (circularized).
The graph model was the edge-centric de Bruijn graph built on the set of all \kpomers of the genome. 
We used $k=31$ and $k=55$ for \ecoli and chr10, respectively, according 
to what has been used in practice for the assembly of such genomes.


We wanted to measure the effect of omnitigs on assembly contiguity in terms of 
(1) increase in contig length, and (2) increase of biological context for elements of interest. 
To measure the increase in length, we measured the average contig length and the E-size.
Since multiple contigs can cover overlapping regions, 
we found the E-size metric~\cite{GAGE} to be more appropriate than the N50 metric.
The E-size of a set of substrings of a genome is defined as the average,
over all genomic positions $i$,
of the mean length of all substrings spanning position $i$. 
This was computed by aligning the contigs to the reference.
\Cref{table:results} shows that omnitigs exhibit significantly more contiguity than unitigs,
with an average contig length that is 62-82\% higher.
There is very little improvement in the E-size (1-4\%), indicating
that most of the improvement come from increasing the length of shorter contigs.

\begin{table}[t]
\caption{Results for $\dbec^{k}(R)$, where $R$ is the set of all \kpomers of the genome. 
\label{table:results}}
\small
\begin{center}
\bgroup \def\arraystretch{1.3} 
\setlength\tabcolsep{3pt} 
	
\begin{tabular}{|l|r|r|r|r|r|r|r|r|}\hline 
	& \multicolumn{4}{|c|}{\ecoli ($k=31$)} & \multicolumn{4}{|c|}{chr10 ($k=55$)} \\
	\hline
& \# strings & avg len & E-size & time (s) & \# strings & avg len & E-size & time (s) \\
\hline \hline
unitigs  & 1,743 &	2,654 &	33,309 & $<1$ & 259,845 &	546 &	8,344 & $1$  \\
Y-to-V   & 1,004 &	4,682 &	33,632 & $<1$ & 159,101 &	878 &	8,376 & $2$ \\
omnitigs & 983 &	4,832 &	34,557 & $<1$ & 158,236 &	887 &	8,401 & $1,046$ \\
\hline 
\end{tabular}
\egroup
\end{center}
\end{table}

\begin{figure}[ht]
\centering
\begin{tabular}{cc}
	\includegraphics[scale=0.45]{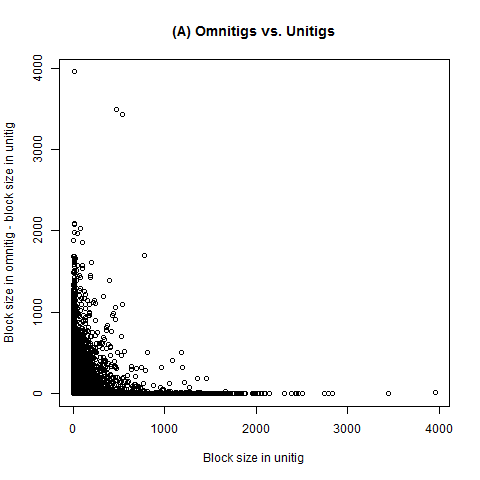} &
	\includegraphics[scale=0.45]{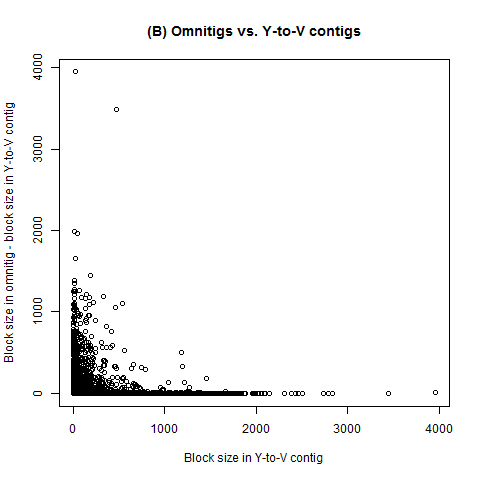} 
\end{tabular}
	\vspace{-7pt}
\caption{
	The increase in SNP block size in omnitigs compared to unitigs (A) and Y-to-V contigs (B).
	Each point is a SNP, and the x-value is the block size of the unitig (in A) or Y-to-V contig (in B) covering it.
	The y-value is the increase in the block size, when compared with omnitigs.
	Note that the y-axis does not represent the block size, but a difference of block sizes.
\label{fig:snps}}
\end{figure}

We wanted to also measure the potential of omnitigs to improve downstream biological analysis, relative to unitigs.
Longer contigs can provide more flanking context around important genomic elements such as SNPs.
One general type of study collects statistics about the relationship of each SNP to other SNPs on the same contig; 
such a study is necessarily limited by the number of SNPs present on the same contig~\cite{Uricaru30012015}. 
We call this number the {\em block size} of a SNP. 
To see the effect of omnitigs on such a study, 
we identified chr10 locations of SNPs in the human population (using dbSNP), and
the block size of each SNP in the omnitig vs. the unitig algorithms.
\Cref{fig:snps}A shows that omnitigs in many cases provide more SNP context.
The number of SNPs whose block size increased was $\sim 1.7$ million  (out of $\sim5.9$ million) and whose
block size increased by more than 10 was $\sim 137$ thousand.
The average number of SNPs per omnitig was 41, with only 26 per unitig. 
Consistent with the contiguity results of~\Cref{table:results}, 
the effect is more pronounced on contigs with less SNPs on them.

We also compared omnitigs to Y-to-V contigs.
Y-to-V contigs have been proposed in the literature~\cite{MGMB07,jacksonthesis,kingsford10},
but, to the best of our knowledge, 
there has not been a quantitative study comparing their merits against other contig algorithms.
Omnitigs also provide more SNP context than Y-to-V contigs, 
with $\sim266$ thousand SNPs having an increase in block size (\Cref{fig:snps}B).
Omnitigs are only marginally better than Y-to-V contigs in terms of average contiguity~(\Cref{table:results}).
Our results suggest that, though not as beneficial as omnitigs, Y-to-V contigs
may nevertheless provide a better alternative to unitigs that is faster than the omnitig algorithm.

\Cref{table:results} also shows the wall-clock running times of our algorithms.
The experiments were run on a node with two Xeon 2.53 GHz CPUs. 
We parallelized the omnitig
algorithm so that it utilized all 8 available cores. 
We observe negligible running times for all algorithms on \ecoli.
On chr10, the running time of the omnitig algorithm is significantly longer (by 18mins) 
than the unitig or Y-to-V algorithm,
though it would still not form a bottleneck in an assembly pipeline.
The memory usage did not exceed 1 GB at any point, 
though we believe it can be significantly reduced with a more careful implementation.

\section{Conclusion}
There are two natural directions for future work: practical and theoretical.
In the practical direction, the omnitig algorithm should be extended to handle the complexities of real data such as sequencing errors, imperfect coverage, 
linear genomes, and double-strandedness.
This is a non-trivial task which is outside the scope of the current study, 
but it will be important in facilitating the application to genome analysis and assembly.
In the theoretical direction, we believe that omnitigs exhibit more structure that can be exploited in a faster algorithm for finding all maximal omnitigs. 
We are also currently studying the graph model which a genomic reconstruction is any collection of circular walks that together cover all nodes/edges of the graph (as in metagenomic sequencing of bacteria). 
We are also studying the class of genome graphs admitting a single safe walk covering all of their nodes or edges, question related to the ones about unique reconstructions.


\section{Acknowledgments}

We would like to thank Daniel Lokshtanov for initial discussions, Rayan Chikhi for feedback on the manuscript, and Nidia Obscura Acosta for very helpful discussion on the manuscript.
This work was supported in part by NSF awards DBI-1356529, IIS-1453527, and IIS-1421908 to PM, and by Academy of Finland grant 274977 to AT.

\section{Author Disclosure Statement}

No competing financial interests exist.

\bibliographystyle{splncs_srt}
\bibliography{arxiv-full-version}

\begin{thebibliography}{10}

\bibitem{spades}
Bankevich, A., Nurk, S., Antipov, D., Gurevich, A.A., Dvorkin, M., Kulikov,
  A.S., Lesin, V.M., Nikolenko, S.I., Pham, S.K., Prjibelski, A.D., Pyshkin,
  A., Sirotkin, A., Vyahhi, N., Tesler, G., Alekseyev, M.A., Pevzner, P.A.:
\newblock {SPAdes}: A new genome assembly algorithm and its applications to
  single-cell sequencing.
\newblock {Journal of Computational Biology} \textbf{19}(5) (2012)  455--477

\bibitem{sspace}
Boetzer, M., Henkel, C.V., Jansen, H.J., Butler, D., Pirovano, W.:
\newblock Scaffolding pre-assembled contigs using {SSPACE}.
\newblock Bioinformatics \textbf{27}(4) (2011)  578--579

\bibitem{boetzer12}
Boetzer, M., Pirovano, W.:
\newblock Toward almost closed genomes with gapfiller.
\newblock Genome Biology \textbf{13}(6) (2012)  1--9

\bibitem{ray}
Boisvert, S., Laviolette, F., Corbeil, J.:
\newblock Ray: simultaneous assembly of reads from a mix of high-throughput
  sequencing technologies.
\newblock Journal of Computational Biology \textbf{17}(11) (2010)  1519--1533

\bibitem{BBT13}
{Bresler}, G., {Bresler}, M., {Tse}, D.:
\newblock {Optimal Assembly for High Throughput Shotgun Sequencing}.
\newblock BMC Bioinformatics \textbf{14}(Suppl 5) (2013)  S18

\bibitem{bcalm}
Chikhi, R., Limasset, A., Jackman, S., Simpson, J.T., Medvedev, P.:
\newblock On the representation of de bruijn graphs.
\newblock In: Research in Computational Molecular Biology, Springer (2014)
  35--55

\bibitem{minia}
Chikhi, R., Rizk, G.:
\newblock {Space-efficient and exact de {B}ruijn graph representation based on
  a Bloom filter}.
\newblock In: WABI. Volume 7534 of Lecture Notes in Computer Science, Springer
  (2012)  236--248

\bibitem{journals/bioinformatics/Guenoche92}
Gu{\'e}noche, A.:
\newblock Can we recover a sequence, just knowing all its subsequences of given
  length?
\newblock Computer Applications in the Biosciences \textbf{8}(6) (1992)
  569--574

\bibitem{genome10k}
Haussler, D., O'Brien, S.J., Ryder, O.A., Barker, F.K., Clamp, M., Crawford,
  A.J., Hanner, R., Hanotte, O., Johnson, W.E., McGuire, J.A.,  et~al.:
\newblock {Genome 10K: a proposal to obtain whole-genome sequence for 10,000
  vertebrate species.}
\newblock {Journal of Heredity} \textbf{100}(6) (2008)  659--674

\bibitem{IW95}
Idury, R.M., Waterman, M.S.:
\newblock A new algorithm for {DNA} sequence assembly.
\newblock Journal of computational biology \textbf{2}(2) (1995)  291--306

\bibitem{SequencingHybridizationLysov1988}
Iu, Florent'ev, V.L., Khorlin, A.A., Khrapko, K.R., Shik, V.V.:
\newblock {Determination of the nucleotide sequence of DNA using hybridization
  with oligonucleotides. A new method}.
\newblock Doklady Akademii nauk SSSR \textbf{303}(6) (1988)  1508--1511

\bibitem{jacksonthesis}
Jackson, B.G.:
\newblock Parallel methods for short read assembly.
\newblock PhD thesis, Iowa State University (2009)

\bibitem{kapun13b}
Kapun, E., Tsarev, F.:
\newblock {De Bruijn} superwalk with multiplicities problem is {NP}-hard.
\newblock BMC bioinformatics \textbf{14}(Suppl 5) (2013) ~S7

\bibitem{kapun13a}
Kapun, E., Tsarev, F.:
\newblock On {NP}-hardness of the paired de {Bruijn} sound cycle problem.
\newblock In: Algorithms in Bioinformatics.
\newblock Springer (2013)  59--69

\bibitem{KM95}
Kececioglu, J.D., Myers, E.W.:
\newblock Combinatiorial algorithms for {DNA} sequence assembly.
\newblock Algorithmica \textbf{13}(1/2) (1995)  7--51

\bibitem{K92}
Kececioglu, J.D.:
\newblock Exact and approximation algorithms for {DNA} sequence reconstruction.
\newblock PhD thesis, University of Arizona, Tucson, AZ, USA (1992)

\bibitem{kingsford10}
Kingsford, C., Schatz, M.C., Pop, M.:
\newblock Assembly complexity of prokaryotic genomes using short reads.
\newblock BMC bioinformatics \textbf{11}(1) (2010) ~21

\bibitem{DBLP:journals/bmcbi/LamKT14}
Lam, K., Khalak, A., Tse, D.:
\newblock Near-optimal assembly for shotgun sequencing with noisy reads.
\newblock {BMC} Bioinformatics \textbf{15}({S-9}) (2014) ~S4

\bibitem{LW88}
Lander, E.S., Waterman, M.S.:
\newblock {Genomic mapping by fingerprinting random clones: a mathematical
  analysis.}
\newblock Genomics \textbf{2}(3) (1988)  231--239

\bibitem{soapdenovo2}
Luo, R., Liu, B., Xie, Y., Li, Z., Huang, W., Yuan, J., He, G., Chen, Y., Pan,
  Q., Liu, Y.,  et~al.:
\newblock Soapdenovo2: an empirically improved memory-efficient short-read de
  novo assembler.
\newblock {GigaScience} \textbf{1}(1) (2012) ~18

\bibitem{MB09}
Medvedev, P., Brudno, M.:
\newblock Maximum likelihood genome assembly.
\newblock Journal of computational biology \textbf{16}(8) (2009)  1101--1116

\bibitem{MGMB07}
Medvedev, P., Georgiou, K., Myers, G., Brudno, M.:
\newblock Computability of models for sequence assembly.
\newblock In: WABI. (2007)  289--301

\bibitem{pdbg}
Medvedev, P., Pham, S., Chaisson, M., Tesler, G., Pevzner, P.:
\newblock Paired de {Bruijn} graphs: a novel approach for incorporating mate
  pair information into genome assemblers.
\newblock Journal of Computational Biology \textbf{18}(11) (2011)  1625--1634

\bibitem{miller_assembly_2010}
Miller, J.R., Koren, S., Sutton, G.:
\newblock Assembly algorithms for next-generation sequencing data.
\newblock Genomics \textbf{95}(6) (2010)  315--327

\bibitem{DBLP:journals/tit/MotahariBT13}
Motahari, A.S., Bresler, G., Tse, D.N.C.:
\newblock Information theory of {DNA} shotgun sequencing.
\newblock {IEEE} Transactions on Information Theory \textbf{59}(10) (2013)
  6273--6289

\bibitem{M05}
Myers, E.W.:
\newblock The fragment assembly string graph.
\newblock In: ECCB/JBI. (2005) ~85

\bibitem{NP09}
Nagarajan, N., Pop, M.:
\newblock Parametric complexity of sequence assembly: Theory and applications
  to next generation sequencing.
\newblock Journal of Computational Biology \textbf{16}(7) (2009)  897--908

\bibitem{nagarajan2013sequence}
Nagarajan, N., Pop, M.:
\newblock Sequence assembly demystified.
\newblock Nature Reviews Genetics \textbf{14}(3) (2013)  157--167

\bibitem{narzisi14}
Narzisi, G., Mishra, B., Schatz, M.C.:
\newblock On algorithmic complexity of biomolecular sequence assembly problem.
\newblock In: Algorithms for Computational Biology.
\newblock Springer (2014)  183--195

\bibitem{peltola83}
Peltola, H., S{\"o}derlund, H., Tarhio, J., Ukkonen, E.:
\newblock Algorithms for some string matching problems arising in molecular
  genetics.
\newblock In: IFIP Congress. (1983)  59--64

\bibitem{P89}
Pevzner, P.A.:
\newblock {L-{T}uple {D}{N}{A} sequencing: computer analysis}.
\newblock J Biomol Struct Dyn \textbf{7}(1) (Aug 1989)  63--73

\bibitem{PTW01}
Pevzner, P.A., Tang, H., Waterman, M.S.:
\newblock An {Eulerian} path approach to {DNA} fragment assembly.
\newblock Proceedings of the National Academy of Sciences \textbf{98}(17)
  (2001)  9748--9753

\bibitem{journals/jcb/RubinovG95}
Rubinov, A.R., Gelfand, M.S.:
\newblock Reconstruction of a string from substring precedence data.
\newblock Journal of Computational Biology \textbf{2}(2) (1995)  371--381

\bibitem{besst}
Sahlin, K., Vezzi, F., Nystedt, B., Lundeberg, J., Arvestad, L.:
\newblock {BESST}-efficient scaffolding of large fragmented assemblies.
\newblock BMC bioinformatics \textbf{15}(1) (2014)  281

\bibitem{salmela15}
Salmela, L., Sahlin, K., M{\"a}kinen, V., Tomescu, A.I.:
\newblock Gap filling as exact path length problem.
\newblock In: Research in Computational Molecular Biology, Springer (2015)
  281--292

\bibitem{GAGE}
Salzberg, S.L., Phillippy, A.M., Zimin, A., Puiu, D., Magoc, T., Koren, S.,
  Treangen, T.J., Schatz, M.C., Delcher, A.L., Roberts, M.:
\newblock {GAGE:} a critical evaluation of genome assemblies and assembly
  algorithms.
\newblock Genome Research (2011)

\bibitem{simpson10}
Simpson, J.T., Durbin, R.:
\newblock {Efficient construction of an assembly string graph using the
  FM-index}.
\newblock Bioinformatics \textbf{26}(12) (2010)  i367--i373

\bibitem{SGA}
Simpson, J.T., Durbin, R.:
\newblock Efficient de novo assembly of large genomes using compressed data
  structures.
\newblock Genome Research (2011)

\bibitem{Uricaru30012015}
Uricaru, R., Rizk, G., Lacroix, V., Quillery, E., Plantard, O., Chikhi, R.,
  Lemaitre, C., Peterlongo, P.:
\newblock {Reference-free detection of isolated SNPs}.
\newblock Nucleic Acids Research \textbf{43}(2) (2015)  e11

\bibitem{vyahhi12}
Vyahhi, N., Pyshkin, A., Pham, S., Pevzner, P.A.:
\newblock From de {Bruijn} graphs to rectangle graphs for genome assembly.
\newblock In: Algorithms in Bioinformatics.
\newblock Springer (2012)  249--261

\bibitem{W95}
Waterman, M.S.:
\newblock Introduction to computational biology: maps, sequences and genomes.
\newblock CRC Press (1995)

\bibitem{velvet}
Zerbino, D.R., Birney, E.:
\newblock {Velvet: algorithms for de novo short read assembly using de Bruijn
  graphs}.
\newblock {Genome Research} \textbf{18}(5) (2008)  821--829

\end{thebibliography}

\end{document}